\documentclass{llncs}
\usepackage{amsmath}
\usepackage{amssymb}
\usepackage{amsfonts}
\usepackage{enumerate}
\usepackage{amssymb}
\usepackage{textcomp}
\usepackage[english]{babel}
\usepackage[pdftex]{graphicx}
\usepackage{hyperref}
\usepackage{textcomp}
\usepackage[hypcap]{caption}
\usepackage[T1]{fontenc}
\usepackage[latin9]{inputenc}
\usepackage{graphicx}
\usepackage{babel}
\usepackage[usenames,dvipsnames]{color}
\usepackage{mathtools}
\usepackage{braket}
\usepackage{placeins}
\usepackage{color, colortbl}
\usepackage{fancyhdr}
\usepackage{float}
\usepackage{algorithm}
\usepackage{algpseudocode}
\usepackage{stmaryrd}
\usepackage{wrapfig}

\numberwithin{equation}{section}

\date{\today}

\newcolumntype{C}{>{\centering\arraybackslash} m{3cm} }

\usepackage[svgnames]{xcolor}
\usepackage{tikzfig}

\tikzstyle{braceedge}=[decorate,decoration={brace,amplitude=10pt}]
\tikzstyle{square box}=[rectangle,fill=white,draw=black,minimum height=6mm,minimum width=6mm,yshift=0.7mm]
\tikzstyle{wire label}=[font=\footnotesize, auto,swap]

\tikzstyle{none}=[inner sep=0pt]

\tikzstyle{opt diredge}=[->]
\tikzstyle{diredge}=[->]
\tikzstyle{cdiag edge}=[-latex]
\tikzstyle{dashed edge}=[dashed]
\tikzstyle{boxedge}=[blue]
\tikzstyle{dashed boxedge}=[boxedge, dashed]

\tikzstyle{dot}=[circle,draw=black,inner sep=1pt,minimum width=2.5mm]
\tikzstyle{white dot}=[dot,fill=white]
\tikzstyle{yellow dot}=[dot,fill=yellow]
\tikzstyle{gray dot}=[dot,fill=gray!60!white]
\tikzstyle{small box}=[rectangle,inner sep=1pt,draw,fill=white,minimum width=3mm,minimum height=3mm,inner sep=1mm]

\tikzstyle{gn}=[circle,fill=Lime,draw=Black]
\tikzstyle{rn}=[circle,fill=Red,draw=Black]
\tikzstyle{H}=[rectangle,fill=Yellow,draw=Black]
\tikzstyle{line}=[scalar,fill=White,draw=Black]
\tikzstyle{io}=[rectangle,fill=White,draw=Black]
\tikzstyle{block}=[rectangle,fill=Orange,draw=Black]
\tikzstyle{graph}=[circle,fill=White,draw=Black]
\tikzstyle{empty}=[rectangle,fill=White,draw=White]
\tikzstyle{box}=[rectangle,fill=White,draw=Black]
\tikzstyle{Dot}=[circle,fill=Black,draw=Black,inner sep=0pt,minimum size=3pt]
\tikzstyle{diam}=[rectangle,fill=Black,draw,yscale=1.2,rotate=45]
\tikzstyle{gangle}=[rectangle,fill=Lime,draw=Black]
\tikzstyle{rangle}=[rectangle,fill=Red,draw=Black]
\tikzstyle{circ}=[circle,fill=none,draw=Black,scale=1.3]
\tikzstyle{ellip}=[ellipse,fill=none,draw=Black,scale=1.3,minimum width =1.3cm]
\tikzstyle{bbox}=[rectangle,fill=Blue,draw=Blue,scale=0.6]
\tikzstyle{gg}=[shape=rectangle,fill=White,draw=Black,dashed]

\tikzstyle{nodev}=[circle,fill=none,draw=Black,scale=1]
\tikzstyle{wirev}=[circle,fill=Black,draw=Black,inner sep=0pt,minimum size=0.5mm]
\tikzstyle{wirevred}=[circle,fill=Red,draw=Black,inner sep=0pt,minimum size=3pt]

\tikzstyle{simple}=[-,draw=Black]
\tikzstyle{directed}=[
  decoration={markings,mark=at position 1 with {\arrow[scale=2]{>}}},postaction={decorate},shorten >=0.4pt]
\tikzstyle{bdirected}=[
  decoration={markings,
  mark=at position 0.01 with {\arrow[scale=-2]{>}},
  mark=at position 0.99 with {\arrow[scale=2]{>}}
  },postaction={decorate},shorten >=0.4pt]
\tikzstyle{bothdirs}=[bdirected,draw=Black]
\tikzstyle{bothdirsred}=[bdirected,draw=Red]
\tikzstyle{blue}=[-,draw=Blue]
\tikzstyle{redd}=[directed,draw=Red]
\tikzstyle{blued}=[directed,draw=Blue]
\tikzstyle{dash}=[dashed,draw=Black]

\tikzstyle{bigpic}=[scale=2.0]

\tikzstyle{every picture}=[baseline=-0.25em,scale=0.5]

\newcommand{\stikz}[2][1]{}
\newcommand{\cstikz}[2][1]{}

\usepackage{xspace}
\usepackage{color}
\def\bR{\begin{color}{red}} 
\def\bB{\begin{color}{blue}}
\def\bM{\begin{color}{magenta}}
\def\bC{\begin{color}{cyan}}
\def\bW{\begin{color}{white}}
\def\bBl{\begin{color}{black}} 
\def\bG{\begin{color}{green}}
\def\bY{\begin{color}{yellow}}
\def\e{\end{color}\xspace}

\usepackage{tikzfig}

\newcommand{\presec}{}
\newcommand{\postsec}{}

\title{Equational reasoning with context-free \\ families of string diagrams
\thanks{The final publication is available at Springer via
\href{http://dx.doi.org/10.1007/978-3-319-21145-9_9}
{http://dx.doi.org/10.1007/978-3-319-21145-9\_9}}
}

\titlerunning{Equational reasoning with context-free families of string diagrams}

\author{Aleks Kissinger \and Vladimir Zamdzhiev}
\authorrunning{A. Kissinger & V. Zamdzhiev}
\institute{University of Oxford\\
\email{\{aleks.kissinger|vladimir.zamdzhiev\}@cs.ox.ac.uk}
}

\begin{document}
\maketitle
\begin{abstract}
  String diagrams provide an intuitive language for expressing networks of
  interacting processes graphically. A discrete representation of string
  diagrams, called string graphs, allows for mechanised equational reasoning by
  double-pushout rewriting. However, one often wishes to express not just single
  equations, but entire families of equations between diagrams of arbitrary
  size. To do this we define a class of context-free grammars, called B-ESG
  grammars, that are suitable for defining entire families of string graphs, and
  crucially, of string graph rewrite rules. We show that the language-membership
  and match-enumeration problems are decidable for these grammars,
  and hence that there is an algorithm for rewriting string graphs according to
  B-ESG rewrite patterns. We also show that it is possible to reason at the
  level of grammars by providing a simple method for transforming a grammar by
  string graph rewriting, and showing admissibility of the induced B-ESG rewrite
  pattern.
\end{abstract}
\presec
\presec
\presec
\section{Introduction}\label{sec:intro}
\postsec
A string diagram (Fig.~\ref{fig:string-diagram-encode}(a)) consists of a
collection of boxes
(typically used to represent certain maps, processes, machines, etc.) connected
by wires.
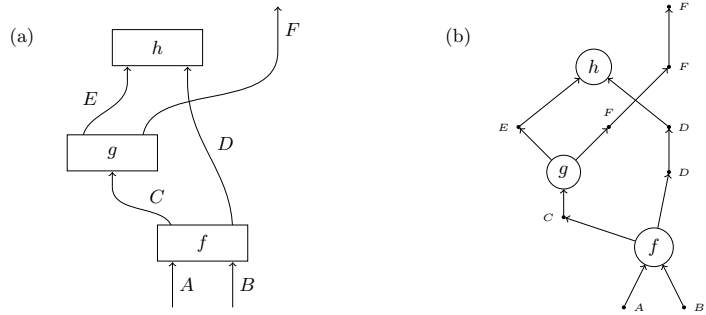
\begin{figure}
  \centering
  \scalebox{0.8}{\begin{tikzpicture}[bigpic]
	\begin{pgfonlayer}{nodelayer}
		\node [style=none] (0) at (-5.25, -2.5) {};
		\node [style=none] (1) at (-4.25, -2.5) {};
		\node [style=wirev] (2) at (2.25, -2.5) {};
		\node [style=none, font={\tiny}] (3) at (2.5, -2.5) {$A$};
		\node [style=wirev] (4) at (3.25, -2.5) {};
		\node [style=none, font={\tiny}] (5) at (3.5, -2.5) {$B$};
		\node [style=none] (6) at (-5.25, -1.75) {};
		\node [style=none] (7) at (-4.25, -1.75) {};
		\node [style=square box, minimum width=1.5 cm] (8) at (-4.75, -1.5) {$f$};
		\node [style=nodev] (9) at (2.75, -1.5) {$f$};
		\node [style=none] (10) at (-5.25, -1.25) {};
		\node [style=none] (11) at (-4.25, -1.25) {};
		\node [style=none, font={\tiny}] (12) at (1, -1) {$C$};
		\node [style=wirev] (13) at (1.25, -1) {};
		\node [style=wirev] (14) at (3, -0.25) {};
		\node [style=none, font={\tiny}] (15) at (3.25, -0.25) {$D$};
		\node [style=none] (16) at (-6.25, -0.5) {};
		\node [style=square box, minimum width=1.5 cm] (17) at (-6.25, 0) {$g$};
		\node [style=none] (18) at (-7.75, 2) {(a)};
		\node [style=nodev] (19) at (1.25, -0.25) {$g$};
		\node [style=none] (20) at (-6.75, 0.25) {};
		\node [style=none] (21) at (-5.75, 0.25) {};
		\node [style=wirev] (22) at (3, 0.5) {};
		\node [style=none, font={\tiny}] (23) at (3.25, 0.5) {$D$};
		\node [style=wirev] (24) at (2, 0.5) {};
		\node [style=none, font={\tiny}] (25) at (0.25, 0.5) {$E$};
		\node [style=wirev] (26) at (0.5, 0.5) {};
		\node [style=none, font={\tiny}] (27) at (2, 0.75) {$F$};
		\node [style=none] (28) at (-6, 1.25) {};
		\node [style=none] (29) at (-5, 1.25) {};
		\node [style=wirev] (30) at (3, 1.5) {};
		\node [style=none, font={\tiny}] (31) at (3.25, 1.5) {$F$};
		\node [style=square box, minimum width=1.5 cm] (32) at (-5.5, 1.75) {$h$};
		\node [style=nodev] (33) at (1.75, 1.5) {$h$};
		\node [style=none] (34) at (-3.5, 1.75) {};
		\node [style=none] (35) at (-3.5, 2.5) {};
		\node [style=wirev] (36) at (3, 2.5) {};
		\node [style=none, font={\tiny}] (37) at (3.25, 2.5) {$F$};
		\node [style=none] (38) at (-5, 1.5) {};
		\node [style=none] (39) at (-6, 1.5) {};
		\node [style=none] (40) at (-6.25, -0.25) {};
		\node [style=none] (41) at (-0.5, 2) {(b)};
	\end{pgfonlayer}
	\begin{pgfonlayer}{edgelayer}
		\draw [style=diredge] (19) to (26);
		\draw [style=simple, in=-90, out=90, looseness=1.00] (11.center) to node[wire label]{$D$} (29.center);
		\draw [style=diredge] (30) to (36);
		\draw [style=diredge] (13) to (19);
		\draw [style=diredge] (9) to (14);
		\draw [style=diredge] (22) to (33);
		\draw [style=simple, in=-90, out=90, looseness=1.00] (21.center) to (34.center);
		\draw [style=simple, in=-90, out=90, looseness=1.00] (10.center) to node[wire label]{$C$} (16.center);
		\draw [style=diredge] (24) to (30);
		\draw [style=diredge] (19) to (24);
		\draw [style=diredge] (2) to (9);
		\draw [style=diredge] (4) to (9);
		\draw [style=diredge] (0.center) to node[wire label]{$A$} (6.center);
		\draw [style=diredge] (26) to (33);
		\draw [style=diredge] (9) to (13);
		\draw [style=simple, in=-90, out=90, looseness=1.00] (20.center) to node[wire label, swap]{$E$} (28.center);
		\draw [style=diredge] (1.center) to node[wire label]{$B$} (7.center);
		\draw [style=diredge] (14) to (22);
		\draw [style=diredge] (34.center) to node[wire label]{$F$} (35.center);
		\draw [style=diredge] (29.center) to (38.center);
		\draw [style=diredge] (28.center) to (39.center);
		\draw [style=diredge] (16.center) to (40.center);
	\end{pgfonlayer}
\end{tikzpicture}}

  \caption{\label{fig:string-diagram-encode} A string diagram and its encoding
  as a string graph}
\end{figure}

They are essentially labelled directed graphs, but with one important
difference: wires, unlike edges, can be left open at one or both ends to form
inputs and outputs, so they have an inherently compositional nature.
Joyal and Street showed in 1991 that compositions of morphisms
in any symmetric monoidal category can be represented using string
diagrams~\cite{joyal_street}, and recently there has been much interest in
applying string diagram-based techniques in a wide variety of fields. In models of
concurrency, they give an elegant presentation of Petri nets with
boundary~\cite{Sobocinski:2010aa}, in computational linguistics, they are used
to compute compositional semantics for sentences~\cite{LambekvsLambek}, and in
control theory, they represent signal-flow
diagrams~\cite{Baez2014a,Bonchi2015}. Equational reasoning for string diagrams
has been used extensively in the program of categorical quantum
mechanics~\cite{cqm}, which provides elegant solutions to problems in quantum
computation, information, and foundations using purely diagrammatic
methods~\cite{picturalism,euler_necessity,NonLocLICS,Coecke:2010aa}.

All of these applications make heavy use of proofs by diagram rewriting. These
are proofs whereby some fixed set of string diagram equations are used to derive
new equations by cutting out the LHS and gluing in the RHS. For example, the
following string diagram equation:
\begin{equation}\label{eq:rule-example}
  \begin{tikzpicture}
	\begin{pgfonlayer}{nodelayer}
		\node [style=none] (0) at (1.25, 1.25) {};
		\node [style=white dot] (1) at (-1.5, -0.75) {};
		\node [style=none] (2) at (-1.5, 1.5) {};
		\node [style=none] (3) at (0, 0) {$=$};
		\node [style=none] (4) at (2.75, 1.5) {};
		\node [style=none] (5) at (1.25, 1.5) {};
		\node [style=none] (6) at (1.25, -1.25) {};
		\node [style=gray dot] (7) at (-1.5, 0.75) {};
		\node [style=gray dot] (8) at (2, -0.5) {};
		\node [style=none] (9) at (-1.5, -1.5) {};
		\node [style=none] (10) at (2.75, 1.25) {};
		\node [style=none] (11) at (2.75, -1.25) {};
		\node [style=white dot] (12) at (2, 0.5) {};
		\node [style=none] (13) at (-3, -1.5) {};
		\node [style=none] (14) at (1.25, -1.5) {};
		\node [style=none] (15) at (-3, 1.5) {};
		\node [style=none] (16) at (2.75, -1.5) {};
		\node [style=gray dot] (17) at (-3, 0.75) {};
		\node [style=white dot] (18) at (-3, -0.75) {};
	\end{pgfonlayer}
	\begin{pgfonlayer}{edgelayer}
		\draw [style=opt diredge] (7) to (2.center);
		\draw [style=opt diredge] (9.center) to (1);
		\draw [style=opt diredge] (17) to (15.center);
		\draw [style=opt diredge] (13.center) to (18);
		\draw [style=opt diredge] (18) to (17);
		\draw [style=opt diredge] (1) to (17);
		\draw [style=opt diredge] (1) to (7);
		\draw [style=opt diredge] (18) to (7);
		\draw [style=opt diredge] (8) to (12);
		\draw [style=swap, in=15, out=-90, looseness=1.00] (10.center) to (12);
		\draw [style=swap, in=165, out=-90, looseness=1.00] (0.center) to (12);
		\draw [style=opt diredge, in=-165, out=90, looseness=1.00] (6.center) to (8);
		\draw [style=opt diredge, in=-15, out=90, looseness=1.00] (11.center) to (8);
		\draw [style=opt diredge] (10.center) to (4.center);
		\draw [style=swap] (6.center) to (14.center);
		\draw (16.center) to (11.center);
		\draw [style=opt diredge] (0.center) to (5.center);
	\end{pgfonlayer}
\end{tikzpicture}
\end{equation}
can be applied to rewrite a larger diagram as follows:
\begin{equation}\label{eq:rewrite-example}
  \begin{tikzpicture}
	\begin{pgfonlayer}{nodelayer}
		\node [style=none] (0) at (-0.75, 1.25) {};
		\node [style=white dot] (1) at (-4.25, -0.75) {};
		\node [style=none] (2) at (-4.25, 2) {};
		\node [style=none] (3) at (0.75, 2) {};
		\node [style=gray dot] (4) at (-0.75, 1.75) {};
		\node [style=none] (5) at (-0.75, -1.25) {};
		\node [style=gray dot] (6) at (-4.25, 0.75) {};
		\node [style=white dot] (7) at (-4.25, -1.75) {};
		\node [style=none] (8) at (0.75, 1.25) {};
		\node [style=none] (9) at (0.75, -1.25) {};
		\node [style=none] (10) at (-5.75, -2) {};
		\node [style=none] (11) at (-0.75, -2) {};
		\node [style=gray dot] (12) at (-5.75, 1.75) {};
		\node [style=white dot] (13) at (0.75, -1.75) {};
		\node [style=gray dot] (14) at (-5.75, 0.75) {};
		\node [style=white dot] (15) at (-5.75, -0.75) {};
		\node [style=none] (16) at (-6.5, -1.25) {};
		\node [style=none] (17) at (-3.5, -1.25) {};
		\node [style=none] (18) at (-6.5, 1.25) {};
		\node [style=none] (19) at (-3.5, 1.25) {};
		\node [style=none] (20) at (-1.5, 1.25) {};
		\node [style=none] (21) at (1.5, -1.25) {};
		\node [style=none] (22) at (-1.5, -1.25) {};
		\node [style=none] (23) at (1.5, 1.25) {};
		\node [style=none] (24) at (-2.5, 0) {$\leadsto$};
		\node [style=none] (25) at (2.5, 0) {$\leadsto$};
		\node [style=none] (26) at (4.25, -1.25) {};
		\node [style=none] (27) at (4.25, 1.25) {};
		\node [style=gray dot] (28) at (4.25, 1.75) {};
		\node [style=none] (29) at (4.25, -2) {};
		\node [style=white dot] (30) at (5, 0.5) {};
		\node [style=white dot] (31) at (5.75, -1.75) {};
		\node [style=none] (32) at (5.75, 2) {};
		\node [style=none] (33) at (5.75, 1.25) {};
		\node [style=gray dot] (34) at (5, -0.5) {};
		\node [style=none] (35) at (6.5, -1.25) {};
		\node [style=none] (36) at (3.5, -1.25) {};
		\node [style=none] (37) at (3.5, 1.25) {};
		\node [style=none] (38) at (5.75, -1.25) {};
		\node [style=none] (39) at (6.5, 1.25) {};
	\end{pgfonlayer}
	\begin{pgfonlayer}{edgelayer}
		\draw [style=opt diredge] (6) to (2.center);
		\draw [style=opt diredge] (7) to (1);
		\draw [style=opt diredge] (14) to (12);
		\draw [style=opt diredge] (10.center) to (15);
		\draw [style=opt diredge] (15) to (14);
		\draw [style=opt diredge] (1) to (14);
		\draw [style=opt diredge] (1) to (6);
		\draw [style=opt diredge] (15) to (6);
		\draw [style=opt diredge] (8.center) to (3.center);
		\draw [style=swap] (5.center) to (11.center);
		\draw [style=dashed edge] (16.center) to (17.center);
		\draw [style=dashed edge] (19.center) to (18.center);
		\draw [style=dashed edge] (17.center) to (19.center);
		\draw [style=dashed edge] (18.center) to (16.center);
		\draw [style=dashed edge] (22.center) to (21.center);
		\draw [style=dashed edge] (23.center) to (20.center);
		\draw [style=dashed edge] (21.center) to (23.center);
		\draw [style=dashed edge] (20.center) to (22.center);
		\draw (13) to (9.center);
		\draw [style=opt diredge] (0.center) to (4);
		\draw [style=opt diredge] (34) to (30);
		\draw [style=swap, in=15, out=-90, looseness=1.00] (33.center) to (30);
		\draw [style=swap, in=165, out=-90, looseness=1.00] (27.center) to (30);
		\draw [style=opt diredge, in=-165, out=90, looseness=1.00] (26.center) to (34);
		\draw [style=opt diredge, in=-15, out=90, looseness=1.00] (38.center) to (34);
		\draw [style=opt diredge] (33.center) to (32.center);
		\draw [style=swap] (26.center) to (29.center);
		\draw [style=dashed edge] (36.center) to (35.center);
		\draw [style=dashed edge] (39.center) to (37.center);
		\draw [style=dashed edge] (35.center) to (39.center);
		\draw [style=dashed edge] (37.center) to (36.center);
		\draw (31) to (38.center);
		\draw [style=opt diredge] (27.center) to (28);
	\end{pgfonlayer}
\end{tikzpicture}
\end{equation}
There are two things to note here. First, the LHS and RHS of string diagram
equations always share a common boundary. In other words, we could think of the
boundary as an invariant sub-diagram that embeds into the
LHS and the RHS of the rule. Secondly, it is this invariant sub-diagram that is
used to glue in the RHS once the LHS is removed. We can formalise this process
using double-pushout (DPO) rewriting.

\begin{figure}[h]
  \centering
  \begin{tikzpicture}
	\begin{pgfonlayer}{nodelayer}
		\node [style=white dot] (0) at (-6.25, -3.5) {};
		\node [style=wirev] (1) at (-6.25, -0.5) {};
		\node [style=gray dot] (2) at (-6.25, -2) {};
		\node [style=white dot] (3) at (-6.25, -5) {};
		\node [style=gray dot] (4) at (-7.75, -0.5) {};
		\node [style=wirev] (5) at (-7.75, -5) {};
		\node [style=gray dot] (6) at (-7.75, -2) {};
		\node [style=white dot] (7) at (-7.75, -3.5) {};
		\node [style=wirev] (8) at (6.25, -4.25) {};
		\node [style=wirev] (9) at (6.25, -1.25) {};
		\node [style=gray dot] (10) at (6.25, -0.5) {};
		\node [style=wirev] (11) at (6.25, -5) {};
		\node [style=white dot] (12) at (7, -2) {};
		\node [style=white dot] (13) at (7.75, -5) {};
		\node [style=wirev] (14) at (7.75, -0.5) {};
		\node [style=wirev] (15) at (7.75, -1.25) {};
		\node [style=gray dot] (16) at (7, -3.5) {};
		\node [style=wirev] (17) at (7.75, -4.25) {};
		\node [style=wirev] (18) at (-6.75, -3) {};
		\node [style=wirev] (19) at (-7.25, -3) {};
		\node [style=wirev] (20) at (-7.75, -2.75) {};
		\node [style=wirev] (21) at (-6.25, -2.75) {};
		\node [style=wirev] (22) at (-7.75, -1.25) {};
		\node [style=wirev] (23) at (-6.25, -1.25) {};
		\node [style=wirev] (24) at (-7.75, -4.25) {};
		\node [style=wirev] (25) at (-6.25, -4.25) {};
		\node [style=wirev] (26) at (0.75, 2.75) {};
		\node [style=wirev] (27) at (0.75, 5.75) {};
		\node [style=wirev] (28) at (-0.75, 5.75) {};
		\node [style=wirev] (29) at (-0.75, 2.75) {};
		\node [style=none] (30) at (1.75, 4.25) {};
		\node [style=none] (31) at (5, 4.25) {};
		\node [style=none] (32) at (1.75, -2.75) {};
		\node [style=none] (33) at (5, -2.75) {};
		\node [style=none] (34) at (-5, -2.75) {};
		\node [style=none] (35) at (-1.75, 4.25) {};
		\node [style=none] (36) at (-1.75, -2.75) {};
		\node [style=none] (37) at (-5, 4.25) {};
		\node [style=none] (38) at (7, 0.25) {};
		\node [style=none] (39) at (7, 2) {};
		\node [style=none] (40) at (0, 0.25) {};
		\node [style=none] (41) at (0, 2) {};
		\node [style=none] (42) at (-7, 0.25) {};
		\node [style=none] (43) at (-7, 2) {};
		\node [style=none] (44) at (5.5, 0) {};
		\node [style=none] (45) at (4.5, 0) {};
		\node [style=none] (46) at (4.5, -1) {};
		\node [style=none] (47) at (-4.5, 0) {};
		\node [style=none] (48) at (-4.5, -1) {};
		\node [style=none] (49) at (-5.5, 0) {};
		\node [style=white dot] (50) at (0.75, -5) {};
		\node [style=wirev] (51) at (-0.75, -5) {};
		\node [style=wirev] (52) at (0.75, -0.5) {};
		\node [style=wirev] (53) at (-0.75, -4.25) {};
		\node [style=wirev] (54) at (0.75, -4.25) {};
		\node [style=wirev] (55) at (-0.75, -1.25) {};
		\node [style=gray dot] (56) at (-0.75, -0.5) {};
		\node [style=wirev] (57) at (0.75, -1.25) {};
		\node [style=wirev] (58) at (7, -2.75) {};
		\node [style=wirev] (59) at (6.25, 2.75) {};
		\node [style=wirev] (60) at (7.75, 2.75) {};
		\node [style=wirev] (61) at (7, 4.25) {};
		\node [style=wirev] (62) at (6.25, 5.75) {};
		\node [style=white dot] (63) at (7, 5) {};
		\node [style=wirev] (64) at (7.75, 5.75) {};
		\node [style=gray dot] (65) at (7, 3.5) {};
		\node [style=white dot] (66) at (-6.25, 3.5) {};
		\node [style=wirev] (67) at (-6.25, 2.75) {};
		\node [style=wirev] (68) at (-7.75, 4.25) {};
		\node [style=white dot] (69) at (-7.75, 3.5) {};
		\node [style=wirev] (70) at (-6.75, 4) {};
		\node [style=wirev] (71) at (-7.75, 5.75) {};
		\node [style=wirev] (72) at (-7.75, 2.75) {};
		\node [style=wirev] (73) at (-6.25, 4.25) {};
		\node [style=wirev] (74) at (-6.25, 5.75) {};
		\node [style=gray dot] (75) at (-7.75, 5) {};
		\node [style=gray dot] (76) at (-6.25, 5) {};
		\node [style=wirev] (77) at (-7.25, 4) {};
	\end{pgfonlayer}
	\begin{pgfonlayer}{edgelayer}
		\draw [style=opt diredge, in=-90, out=15, looseness=1.00] (12) to (15);
		\draw [style=opt diredge, in=-90, out=165, looseness=1.00] (12) to (9);
		\draw [style=opt diredge, in=-165, out=90, looseness=1.00] (8) to (16);
		\draw [style=opt diredge, in=-15, out=90, looseness=1.00] (17) to (16);
		\draw [style=opt diredge] (15) to (14);
		\draw [style=opt diredge] (11) to (8);
		\draw [style=opt diredge] (13) to (17);
		\draw [style=opt diredge] (9) to (10);
		\draw [-latex] (30.center) to (31.center);
		\draw [-latex] (32.center) to (33.center);
		\draw [-latex] (35.center) to (37.center);
		\draw [-latex] (36.center) to (34.center);
		\draw [-latex] (39.center) to (38.center);
		\draw [-latex] (41.center) to (40.center);
		\draw [-latex] (43.center) to (42.center);
		\draw (44.center) to (45.center);
		\draw (46.center) to (45.center);
		\draw (49.center) to (47.center);
		\draw (48.center) to (47.center);
		\draw [style=opt diredge] (57) to (52);
		\draw [style=opt diredge] (51) to (53);
		\draw [style=opt diredge] (50) to (54);
		\draw [style=opt diredge] (55) to (56);
		\draw [style=opt diredge] (16) to (58);
		\draw [style=opt diredge] (58) to (12);
		\draw [style=opt diredge, in=-90, out=15, looseness=1.00] (63) to (64);
		\draw [style=opt diredge, in=-90, out=165, looseness=1.00] (63) to (62);
		\draw [style=opt diredge, in=-165, out=90, looseness=1.00] (59) to (65);
		\draw [style=opt diredge, in=-15, out=90, looseness=1.00] (60) to (65);
		\draw [style=opt diredge] (65) to (61);
		\draw [style=opt diredge] (61) to (63);
		\draw [style=opt diredge] (0) to (18);
		\draw [style=opt diredge] (18) to (6);
		\draw [style=opt diredge] (7) to (19);
		\draw [style=opt diredge] (19) to (2);
		\draw [style=opt diredge] (0) to (21);
		\draw [style=opt diredge] (21) to (2);
		\draw [style=opt diredge] (7) to (20);
		\draw [style=opt diredge] (20) to (6);
		\draw [style=opt diredge] (5) to (24);
		\draw [style=opt diredge] (24) to (7);
		\draw [style=opt diredge] (3) to (25);
		\draw [style=opt diredge] (25) to (0);
		\draw [style=opt diredge] (6) to (22);
		\draw [style=opt diredge] (22) to (4);
		\draw [style=opt diredge] (2) to (23);
		\draw [style=opt diredge] (23) to (1);
		\draw [style=opt diredge] (66) to (70);
		\draw [style=opt diredge] (70) to (75);
		\draw [style=opt diredge] (69) to (77);
		\draw [style=opt diredge] (77) to (76);
		\draw [style=opt diredge] (66) to (73);
		\draw [style=opt diredge] (73) to (76);
		\draw [style=opt diredge] (69) to (68);
		\draw [style=opt diredge] (68) to (75);
		\draw [style=opt diredge] (72) to (69);
		\draw [style=opt diredge] (67) to (66);
		\draw [style=opt diredge] (75) to (71);
		\draw [style=opt diredge] (76) to (74);
	\end{pgfonlayer}
\end{tikzpicture}
  \caption{DPO example}\label{fig:dpo}
\end{figure}
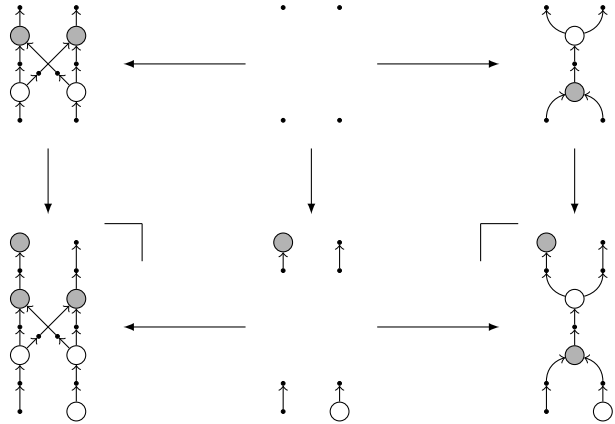
We begin by representing (directed) string diagrams as certain labelled,
(directed) graphs called \textit{string graphs} (see
Fig.~\ref{fig:string-diagram-encode}). Wires are replaced by chains of edges
containing special dummy
vertices called \textit{wire-vertices}. By contrast, the `real' vertices,
labelled here $f, g$ and $h$ are called \textit{node-vertices}. String graphs---
originally introduced under the name `open graphs' in~\cite{open_graphs1}---have
the advantage that they are purely combinatoric objects, as opposed to the
geometric objects like string diagrams. As such, they form a suitable category
for performing DPO rewriting. The DPO diagram associated with the
rewrite~\eqref{eq:rewrite-example} is given in Figure~\ref{fig:dpo},
where the top row is the string graph rewrite rule for~\eqref{eq:rule-example},
the left square is the pushout complement
removing the LHS, and the right square is the pushout gluing in the RHS.

This technique has been used to mechanise proofs involving string diagrams, and
forms the foundation of the diagrammatic proof assistant
Quantomatic~\cite{quanto-cade}. However, proving equations between single string
diagrams is just half the story. Typically, one wishes to prove properties about
entire families of diagrams. For example, suppose we have a node that serves as
an $n$-fold `copy' process. Then, one might require, as in
Fig.~\ref{fig:n-copies}(a), that connecting another node to the `copy' node
would result in $n$ copies. Thus, we have an infinite family of (very similar) equations,
one for each $n$.

\begin{figure}
  \centering
  \begin{tikzpicture}
	\begin{pgfonlayer}{nodelayer}
		\node [style=white dot] (0) at (4.25, -0.25) {};
		\node [style=gray dot] (1) at (4.25, -1.75) {};
		\node [style=wirev] (2) at (4.25, 1) {};
		\node [style=none] (3) at (5.75, 0) {$=$};
		\node [style=wirev] (4) at (7.5, 1) {};
		\node [style=gray dot] (5) at (7.5, 0.25) {};
		\node [style=wirev] (6) at (4.25, -1) {};
		\node [style=bbox] (7) at (3.75, 1.5) {};
		\node [style=none] (8) at (4.75, 1.5) {};
		\node [style=none] (9) at (4.75, 0.5) {};
		\node [style=none] (10) at (3.75, 0.5) {};
		\node [style=bbox] (11) at (7, 1.5) {};
		\node [style=none] (12) at (8, 1.5) {};
		\node [style=none] (13) at (8, -0.25) {};
		\node [style=none] (14) at (7, -0.25) {};
		\node [style=none] (15) at (-8, 0.75) {(a)};
		\node [style=none] (16) at (2.5, 0.75) {(b)};
		\node [style=none] (17) at (-1.5, 1.25) {};
		\node [style=none] (18) at (-3, 1.25) {};
		\node [style=gray dot] (19) at (-3, 0) {};
		\node [style=gray dot] (20) at (-1.5, 0) {};
		\node [style=none] (21) at (-4.25, 0) {$=$};
		\node [style=none] (22) at (-5.25, 1.25) {};
		\node [style=none] (23) at (-6, 1) {...};
		\node [style=white dot] (24) at (-6, 0) {};
		\node [style=none] (25) at (-2.25, 0.75) {...};
		\node [style=gray dot] (26) at (-6, -1) {};
		\node [style=none] (27) at (-6.75, 1.25) {};
	\end{pgfonlayer}
	\begin{pgfonlayer}{edgelayer}
		\draw [style=diredge] (0) to (2);
		\draw [style=diredge] (5) to (4);
		\draw [style=diredge] (1) to (6);
		\draw [style=diredge] (6) to (0);
		\draw [style=boxedge] (7) to (8.center);
		\draw [style=boxedge] (8.center) to (9.center);
		\draw [style=boxedge] (9.center) to (10.center);
		\draw [style=boxedge] (10.center) to (7);
		\draw [style=boxedge] (11) to (12.center);
		\draw [style=boxedge] (12.center) to (13.center);
		\draw [style=boxedge] (13.center) to (14.center);
		\draw [style=boxedge] (14.center) to (11);
		\draw [style=diredge] (24) to (27.center);
		\draw [style=diredge] (24) to (22.center);
		\draw [style=diredge] (26) to (24);
		\draw [style=diredge] (19) to (18.center);
		\draw [style=diredge] (20) to (17.center);
	\end{pgfonlayer}
\end{tikzpicture}

  \caption{\label{fig:n-copies} An $n$-fold copy rule, and its formalisation
  using !-box notation.}
\end{figure}
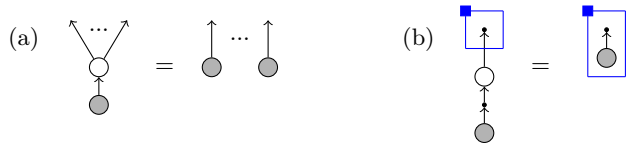

Such a family of equations can be easily captured using a graphical syntax called
\textit{!-box notation}. Here, we can indicate that a subgraph (along with its
adjacent edges) on the LHS and RHS of a rule can be repeated any number of times
by wrapping a box around it, as in Fig.~\ref{fig:n-copies}(b). This rewrite pattern
can then be instantiated to a concrete rewrite rule by fixing the number of copies
of each !-box to retain. A formal description of this instantiation process is given
in~\cite{pattern_graphs}. This procedure is straightforward to mechanise, and is
the main mechanism Quantomatic uses for reasoning about string diagram families.
However, the types of string graph languages representable using !-box notation is
quite limited. For example, the languages produced by string graphs with
!-boxes always have finitely-bounded diameter and chromatic number.\footnote{For colourability and cliques in string diagrams, we treat chains of wire-vertices as single edges.} Thus many
naturally-occurring languages containing chains or cliques of arbitrary size are
not possible.

In~\cite{ContextFreeGAM}, we showed that the languages generated by string
graphs with (non-overlapping) !-boxes can always be generated using a context-free
vertex replacement grammar. Thus, we conjectured that a general notion of a
context-free grammar for string graphs and string graph rewrite rules could
produce a more expressive language for reasoning about families of string
diagrams. In this paper, we will show that this is indeed the case.

We begin by defining a context-free grammar for string graphs which is built on
the well-known B-edNCE class of grammars. We call these grammars
\textit{boundary encoded string graph} (B-ESG) grammars, where `boundary' here
means the grammar satisfies the same boundary condition as B-edNCE grammars. An
encoded string graph is a slight generalization of a string graph, where certain
fixed subgraphs can be encoded as a single edge with a special label. A B-ESG
grammar consists of a B-edNCE grammar for producing encoded string graphs, and a
set of decoding rules for replacing these special edges. For example, this
grammar:
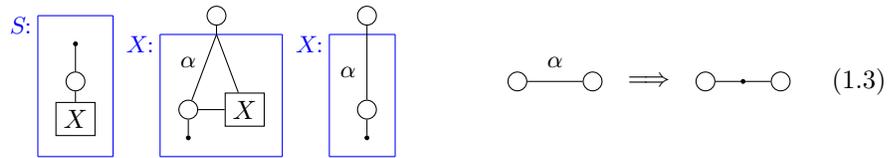
\begin{equation}\label{eq:complete-example}
\begin{tikzpicture}
	\begin{pgfonlayer}{nodelayer}
		\node [style=white dot] (0) at (2.5, 0) {};
		\node [style=white dot] (1) at (0.5, 0) {};
		\node [style=none, font={\small}] (2) at (1.5, 0.5) {$\alpha$};
		\node [style=none] (3) at (4, 0) {$\Longrightarrow$};
		\node [style=white dot] (4) at (7.5, 0) {};
		\node [style=white dot] (5) at (5.5, 0) {};
		\node [style=wirev] (6) at (6.5, 0) {};
		\node [style=wirev] (7) at (-8.25, -1.5) {};
		\node [style=small box] (8) at (-6.75, -0.75) {$X$};
		\node [style=none] (9) at (-10.25, -2) {};
		\node [style=none] (10) at (-5.75, -2) {};
		\node [style=none] (11) at (-10.25, 1.75) {};
		\node [style=none] (12) at (-3.5, 1.25) {};
		\node [style=none] (13) at (-4.5, -2) {};
		\node [style=none] (14) at (-9, -2) {};
		\node [style=wirev] (15) at (-3.5, -1.5) {};
		\node [style=none] (16) at (-9, 1.25) {};
		\node [style=none, font={\small\color{blue}}, anchor=east, xshift=0.5mm] (17) at (-9.25, 1) {$X\!\!:$};
		\node [style=white dot] (18) at (-3.5, -0.75) {};
		\node [style=none] (19) at (-2.75, -2) {};
		\node [style=white dot] (20) at (-11.25, 0) {};
		\node [style=small box] (21) at (-11.25, -1) {$X$};
		\node [style=none] (22) at (-12.25, 1.75) {};
		\node [style=none, font={\small}] (23) at (-8.25, 0.5) {$\alpha$};
		\node [style=none] (24) at (-2.75, 1.25) {};
		\node [style=none] (25) at (-5.75, 1.25) {};
		\node [style=none] (26) at (-7.5, 1.25) {};
		\node [style=wirev] (27) at (-11.25, 1) {};
		\node [style=white dot] (28) at (-8.25, -0.75) {};
		\node [style=none, font={\small}] (29) at (-4, 0.25) {$\alpha$};
		\node [style=none, font={\small\color{blue}}, anchor=east, xshift=0.5mm] (30) at (-4.75, 1) {$X\!\!:$};
		\node [style=none, font={\small\color{blue}}, anchor=east, xshift=0.5mm] (31) at (-12.5, 1.5) {$S\!\!:$};
		\node [style=white dot] (32) at (-7.5, 1.75) {};
		\node [style=white dot] (33) at (-3.5, 1.75) {};
		\node [style=none] (34) at (-4.5, 1.25) {};
		\node [style=none] (35) at (-12.25, -2) {};
	\end{pgfonlayer}
	\begin{pgfonlayer}{edgelayer}
		\draw (1) to (0);
		\draw (5) to (6);
		\draw (6) to (4);
		\draw [style=boxedge] (16.center) to (25.center);
		\draw [style=boxedge] (25.center) to (10.center);
		\draw [style=boxedge] (10.center) to (14.center);
		\draw [style=boxedge] (14.center) to (16.center);
		\draw (32) to (26.center);
		\draw (26.center) to (28);
		\draw (28) to (7);
		\draw (28) to (8);
		\draw (26.center) to (8);
		\draw [style=boxedge] (22.center) to (11.center);
		\draw [style=boxedge] (11.center) to (9.center);
		\draw [style=boxedge] (9.center) to (35.center);
		\draw [style=boxedge] (35.center) to (22.center);
		\draw (27) to (20);
		\draw (21) to (20);
		\draw [style=boxedge] (34.center) to (24.center);
		\draw [style=boxedge] (24.center) to (19.center);
		\draw [style=boxedge] (19.center) to (13.center);
		\draw [style=boxedge] (13.center) to (34.center);
		\draw (33) to (12.center);
		\draw (18) to (15);
		\draw (18) to (12.center);
	\end{pgfonlayer}
\end{tikzpicture}
\end{equation}
produces a complete graph of $n \geq 2$ white node-vertices, connected by wires,
which would not be possible in the !-box language. We define B-ESG rules analogously
as two grammars with identical non-terminals, and a 1-to-1 correspondence between their
productions. Consider, for example, the following rule:
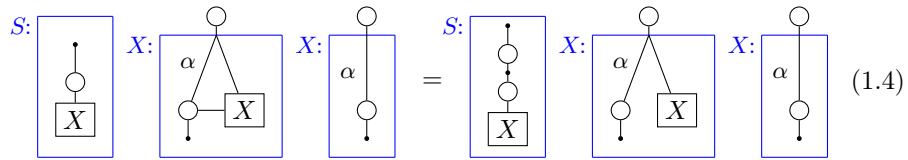
\begin{equation}\label{eq:complete-rule-example}
  \begin{tikzpicture}
	\begin{pgfonlayer}{nodelayer}
		\node [style=none] (0) at (-8.5, 1.25) {};
		\node [style=none] (1) at (-5.25, 1.25) {};
		\node [style=none] (2) at (-5.25, -2) {};
		\node [style=none] (3) at (-8.5, -2) {};
		\node [style=none, font={\small\color{blue}}, anchor=east, xshift=0.5mm] (4) at (-8.75, 1) {$X\!\!:$};
		\node [style=small box] (5) at (-6.25, -0.75) {$X$};
		\node [style=none] (6) at (-7, 1.25) {};
		\node [style=white dot] (7) at (-7.75, -0.75) {};
		\node [style=wirev] (8) at (-7.75, -1.5) {};
		\node [style=white dot] (9) at (-7, 1.75) {};
		\node [style=none] (10) at (-1.25, 0) {$=$};
		\node [style=none] (11) at (-9.75, -2) {};
		\node [style=none] (12) at (-9.75, 1.75) {};
		\node [style=none, font={\small\color{blue}}, anchor=east, xshift=0.5mm] (13) at (-12, 1.5) {$S\!\!:$};
		\node [style=none] (14) at (-11.75, -2) {};
		\node [style=none] (15) at (-11.75, 1.75) {};
		\node [style=wirev] (16) at (-10.75, 1) {};
		\node [style=white dot] (17) at (-10.75, 0) {};
		\node [style=small box] (18) at (-10.75, -1) {$X$};
		\node [style=none] (19) at (-2.25, -2) {};
		\node [style=none] (20) at (-2.25, 1.25) {};
		\node [style=none] (21) at (-4, -2) {};
		\node [style=none] (22) at (-4, 1.25) {};
		\node [style=none] (23) at (1.75, 1.75) {};
		\node [style=none] (24) at (-0.25, 1.75) {};
		\node [style=small box] (25) at (0.75, -1.25) {$X$};
		\node [style=none] (26) at (1.75, -2) {};
		\node [style=none] (27) at (-0.25, -2) {};
		\node [style=white dot] (28) at (0.75, -0.25) {};
		\node [style=none, font={\small\color{blue}}, anchor=east, xshift=0.5mm] (29) at (-4.25, 1) {$X\!\!:$};
		\node [style=none] (30) at (-3, 1.25) {};
		\node [style=white dot] (31) at (-3, 1.75) {};
		\node [style=none, font={\small\color{blue}}, anchor=east, xshift=0.5mm] (32) at (-0.5, 1.5) {$S\!\!:$};
		\node [style=none, font={\small}] (33) at (-7.75, 0.5) {$\alpha$};
		\node [style=wirev] (34) at (0.75, 1.5) {};
		\node [style=wirev] (35) at (3.75, -1.5) {};
		\node [style=none] (36) at (6.25, -2) {};
		\node [style=small box] (37) at (5.25, -0.75) {$X$};
		\node [style=none, font={\small}] (38) at (3.75, 0.5) {$\alpha$};
		\node [style=none] (39) at (4.5, 1.25) {};
		\node [style=none] (40) at (6.25, 1.25) {};
		\node [style=none] (41) at (3, -2) {};
		\node [style=white dot] (42) at (3.75, -0.75) {};
		\node [style=none] (43) at (3, 1.25) {};
		\node [style=white dot] (44) at (4.5, 1.75) {};
		\node [style=none, font={\small\color{blue}}, anchor=east, xshift=0.5mm] (45) at (2.75, 1) {$X\!\!:$};
		\node [style=white dot] (46) at (0.75, 0.75) {};
		\node [style=wirev] (47) at (0.75, 0.25) {};
		\node [style=wirev] (48) at (-3, -1.5) {};
		\node [style=white dot] (49) at (-3, -0.75) {};
		\node [style=none, font={\small}] (50) at (-3.5, 0.25) {$\alpha$};
		\node [style=none] (51) at (9.25, -2) {};
		\node [style=none] (52) at (8.5, 1.25) {};
		\node [style=none] (53) at (9.25, 1.25) {};
		\node [style=none] (54) at (7.5, -2) {};
		\node [style=none, font={\small\color{blue}}, anchor=east, xshift=0.5mm] (55) at (7.25, 1) {$X\!\!:$};
		\node [style=none, font={\small}] (56) at (8, 0.25) {$\alpha$};
		\node [style=wirev] (57) at (8.5, -1.5) {};
		\node [style=none] (58) at (7.5, 1.25) {};
		\node [style=white dot] (59) at (8.5, 1.75) {};
		\node [style=white dot] (60) at (8.5, -0.75) {};
	\end{pgfonlayer}
	\begin{pgfonlayer}{edgelayer}
		\draw [style=boxedge] (0.center) to (1.center);
		\draw [style=boxedge] (1.center) to (2.center);
		\draw [style=boxedge] (2.center) to (3.center);
		\draw [style=boxedge] (3.center) to (0.center);
		\draw (9) to (6.center);
		\draw (6.center) to (7);
		\draw (7) to (8);
		\draw (7) to (5);
		\draw (6.center) to (5);
		\draw [style=boxedge] (15.center) to (12.center);
		\draw [style=boxedge] (12.center) to (11.center);
		\draw [style=boxedge] (11.center) to (14.center);
		\draw [style=boxedge] (14.center) to (15.center);
		\draw (16) to (17);
		\draw (18) to (17);
		\draw [style=boxedge] (22.center) to (20.center);
		\draw [style=boxedge] (20.center) to (19.center);
		\draw [style=boxedge] (19.center) to (21.center);
		\draw [style=boxedge] (21.center) to (22.center);
		\draw [style=boxedge] (24.center) to (23.center);
		\draw [style=boxedge] (23.center) to (26.center);
		\draw [style=boxedge] (26.center) to (27.center);
		\draw [style=boxedge] (27.center) to (24.center);
		\draw (25) to (28);
		\draw (31) to (30.center);
		\draw (28) to (34);
		\draw [style=boxedge] (43.center) to (40.center);
		\draw [style=boxedge] (40.center) to (36.center);
		\draw [style=boxedge] (36.center) to (41.center);
		\draw [style=boxedge] (41.center) to (43.center);
		\draw (44) to (39.center);
		\draw (39.center) to (42);
		\draw (42) to (35);
		\draw (39.center) to (37);
		\draw (49) to (48);
		\draw (49) to (30.center);
		\draw [style=boxedge] (58.center) to (53.center);
		\draw [style=boxedge] (53.center) to (51.center);
		\draw [style=boxedge] (51.center) to (54.center);
		\draw [style=boxedge] (54.center) to (58.center);
		\draw (59) to (52.center);
		\draw (60) to (57);
		\draw (60) to (52.center);
	\end{pgfonlayer}
\end{tikzpicture}
\end{equation}
It will rewrite any complete graph of white node-vertices into a star with white node-vertices on every outgoing wire:
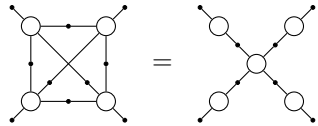
\ctikzfig{complete-ex-inst}
In this paper, we will define a family of context-free grammars suitable for equational reasoning with string diagrams. After giving formal definitions for string graphs in Section~\ref{sec:prelims}, we will define encoded string diagrams and B-ESG grammars in Section~\ref{sec:besg}, as well as prove that a B-ESG grammar always yields a language consisting of well-formed string graphs. We also show that the B-ESG membership problem is decidable as well as the match enumeration problem. In Section~\ref{sec:rewrite} we will extend to B-ESG rewrite rules and show that these always form a language of well-formed string graph rewrite rules (i.e. the generated rules always have corresponding inputs/outputs). In Section~\ref{sec:transforming}, we give a simple example of meta-level reasoning with B-ESG grammars, whereby string graph rewrite rules can be lifted to admissible transformations of B-ESG grammars, thus proving entire families of diagram equations simultaneously. We then generalise this result to show how basic B-ESG on B-ESG rewriting is possible. In the conclusion, we discuss how these principles might be extended to more powerful versions of B-ESG on B-ESG rewriting and structural induction.
\presec
\section{Preliminaries}\label{sec:prelims}
\postsec
\begin{definition}[Graph \cite{graph_grammar_handbook}] \rm
  A \textit{graph} over an alphabet of vertex labels $\Sigma$ and an alphabet of edge
  labels $\Gamma$ is a tuple $H = (V, E, \lambda)$, where $V$ is a finite set
  of vertices, $E \subseteq \{(v, \gamma, w) | v, w \in V, v \not= w, \gamma \in
  \Gamma\}$ is the set of edges and $\lambda : V \to \Sigma$ is the node
  labelling function. The set of all graphs with labels $\Sigma, \Gamma$ is denoted $GR_{\Sigma,\Gamma}$.
\end{definition}
Note that this notion of graphs forbids self-loops. This is a standard (and convenient) assumption in the vertex-replacement grammar literature. It will not get in our way, since we always use chains of wire-vertices to encode self-loops in string diagrams. Also, this notion of graph allows parallel edges only if they have different types. Again, this isn't problematic for our use case, because string graphs cannot have parallel edges (but they do allow parallel \textit{wires}, cf. Definition~\ref{def:wire}).
\presec
\subsection{B-edNCE grammars}
\postsec
We will focus on neighbourhood-controlled embedding (NCE) grammars. This is a type of graph grammar where non-terminal vertices are replaced by graphs according to a set of \textit{productions}, each endowed with a set of \textit{connection instructions} which determine how the new graph should be connected to the neighbourhood of the non-terminal. edNCE grammars are NCE grammars, with \textbf{e}dge labels and \textbf{d}irections, and B-edNCE grammars additionally impose the `boundary' condition~\cite{BoundaryGrammar} which guarantees confluence for applications of productions. They form an important subclass of all confluent edNCE grammars (C-edNCE grammars) that is particularly easy to characterise.
\begin{definition}[Graph with embedding \cite{graph_grammar_handbook}] \rm
  A \textit{graph with embedding} over labels $\Sigma, \Gamma$ is a pair $(H,
  C),$ where $H$ is a graph over $\Sigma, \Gamma$ and $C \subseteq \Sigma
  \times \Gamma \times \Gamma \times V_H \times \{in, out\}$. $C$ is called a \textit{connection relation} and its elements are called \textit{connection instructions}. The set of all graphs with embedding over $\Sigma, \Gamma$ is denoted by $GRE_{\Sigma, \Gamma}$.
\end{definition}
Graph grammars operate by substituting a graph (with embedding) for a non-terminal vertex of another graph. Connection instructions are used to introduce edges connected to the new graph based on edges connected to the non-terminal. A connection instruction $(\sigma, \beta / \gamma, x, d)$ says to add an edge labelled $\gamma$ connected to the vertex $x$ in the new graph, for every $\beta$-labelled edge connecting a $\sigma$-labelled vertex to the non-terminal. $d$ then indicates whether this rule applies to in-edges or out-edges of the non-terminal. For the formal definition of this substitution operation, see e.g.~\cite{graph_grammar_handbook}.
\begin{definition}[edNCE Graph Grammar \cite{graph_grammar_handbook}] \rm
  An \textit{edNCE Graph Grammar} is a tuple $G = (\Sigma, \Delta, \Gamma,
  \Omega, P, S)$, where 
    $\Sigma$ is the alphabet of vertex labels,
    $\Delta \subseteq \Sigma$ is the alphabet of terminal vertex labels,
    $\Gamma$ is the alphabet of edge labels,
    $\Omega \subseteq \Gamma$ is the alphabet of final edge labels,
    $P$ is a finite set of productions and
    $S \in \Sigma - \Delta$ is the initial nonterminal.
  Productions are of the form $X \rightarrow (D, C)$, where $X \in \Sigma -
  \Delta$ is a non-terminal vertex and $(D, C) \in GRE_{\Sigma, \Gamma}$ is a
  graph with embedding.
\end{definition}
A \textit{derivation} in an edNCE grammar is a sequence of substitutions of non-terminals starting from the the graph which just contains the single starting non-terminal $S$. The set of all graphs (not containing non-terminals) isomorphic to some graph reachable in this manner is called the \textit{language} of the grammar. edNCE grammars with the additional property that the order in which non-terminals are expanded is irrelevant are called \textit{confluent} edNCE, or C-edNCE, grammars. We will focus on a special case:
\begin{definition}[B-edNCE grammar] \rm
  A \textit{boundary edNCE}, or B-edNCE, grammar is a grammar such that for all productions $p : X \to (D,C)$, $D$ contains no adjacent non-terminal nodes and $C$ contains no connection instructions of the form $(\sigma, \beta/\gamma, x, d)$ where $\sigma$ is a non-terminal label.
\end{definition}
\presec
\subsection{String graphs and rewriting}
\postsec
\begin{definition}[String Graph] \label{def:string-graph} \rm
  For disjoint sets $\mathcal N = \{ N_f, N_g, \ldots \}$, $\mathcal W = \{ W_A, W_B, \ldots \}$, a (directed) \textit{string graph} is a directed graph labelled by the set $\mathcal N \cup \mathcal W$, where vertices with labels in $\mathcal N$ are called \textit{node-vertices} and vertices with labels in $\mathcal W$ are called \textit{wire-vertices}, and the following conditions hold:
    (1) there are no edges directly connecting two node-vertices
    (2) the in-degree of every wire-vertex is at most one and
    (3) the out-degree of every wire-vertex is at most one.
\end{definition}
The category \textbf{SGraph} has as its objects string graphs and its morphisms string graph homomorphisms (i.e.~graph homomorphisms respecting labels). In full generality, string graphs also allow one to restrict which node-vertices can be connected to which wire-vertices, and allow for an ordering on in- and out-edges (e.g. for non-commutative maps), but for simplicity, we will consider the case where any node-vertex can be connected to any wire-vertex and the ordering is irrelevant.

We will depict wire-vertices as small black dots and node-vertices as larger nodes of various shapes and colours. We can define undirected string graphs analogously, by replacing the last two conditions with the requirement that each wire-vertex have degree at most 2. To avoid excessive duplication we will state all of our results for the directed case, but similar results carry through to the undirected case. Thus, we will occasionally give undirected examples when they are more convenient than their directed counterparts.
\begin{definition}[Inputs, Outputs and Boundary Vertices] \rm
A wire-vertex of a string graph $G$ is called an \textit{input} if it has no incoming
edges. A wire-vertex with no outgoing edges is called an \textit{output}.
The boundary of $G$ consists of all of its inputs and outputs.
\end{definition}
Wires in string diagrams are geometric in nature. They are encoded in string graphs as chains of wire-vertices.
\begin{definition}\label{def:wire} \rm
  A \textit{wire} is a maximal connected subgraph of a string graph consisting of only wire-vertices and at least one edge. There are three cases: (a) it forms a simple directed cycle, which is called a \textit{circle}, (b) it is a chain where one or both endpoints are connected to node-vertices, which is called an \textit{attached wire}, or (c) it is a chain not connected to any node-verties, which is called a \textit{bare wire}.
\end{definition}
\begin{figure}
  \centering

  \begin{tikzpicture}
	\begin{pgfonlayer}{nodelayer}
		\node [style=wirev] (0) at (-6.25, -0.75) {};
		\node [style=wirev] (1) at (-4.75, -0.75) {};
		\node [style=wirev] (2) at (-5.5, 0.75) {};
		\node [style=white dot] (3) at (0, 1) {};
		\node [style=wirev] (4) at (0, 0) {};
		\node [style=wirev] (5) at (0, -1) {};
		\node [style=wirev] (6) at (5.5, -1) {};
		\node [style=wirev] (7) at (5.5, 0) {};
		\node [style=wirev] (8) at (5.5, 1) {};
		\node [style=none] (9) at (-7.25, 0.75) {(a)};
		\node [style=none] (10) at (-1.5, 0.75) {(b)};
		\node [style=none] (11) at (4, 0.75) {(c)};
	\end{pgfonlayer}
	\begin{pgfonlayer}{edgelayer}
		\draw [style=diredge, in=-135, out=-45, looseness=1.00] (0) to (1);
		\draw [style=diredge, in=0, out=45, looseness=1.00] (1) to (2);
		\draw [style=diredge, in=135, out=180, looseness=1.00] (2) to (0);
		\draw [style=diredge] (5) to (4);
		\draw [style=diredge] (4) to (3);
		\draw [style=diredge] (6) to (7);
		\draw [style=diredge] (7) to (8);
	\end{pgfonlayer}
\end{tikzpicture}

  \caption{\label{fig:wire-cases} Types of wires: (a) circle, (b) attached wire, (c) bare wire}
\end{figure}
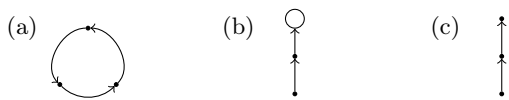
In particular, when considering embeddings of string diagrams into each other, wires are allowed to be sub-divided arbitrarily. To accommodate this behaviour with string graphs, we represent wires as chains of wire-vertices. Of course, this choice is not unique, as we can represent a single wire with a chain of wire-vertices of any length. Being isomorphic up to the number of wire-vertices representing each wire is a natural notion of `sameness' for string diagrams, which is called \textit{wire-homeomorphism}.
\begin{definition}[Wire-homeomorphic string graphs] \label{def:wire-homeo} \rm
  Two string graphs $G$ and $G'$ are called \textit{wire-homeomorphic}, written $G \sim G'$ if $G'$ can be obtained from $G$ by either merging two adjacent wire-vertices (left) or by splitting a wire-vertex into two adjacent wire-vertices (right) any number of times:
    \[ \begin{tikzpicture}
	\begin{pgfonlayer}{nodelayer}
		\node [style=none] (0) at (-1.5, 0) {};
		\node [style=wirev] (1) at (-0.5, 0) {};
		\node [style=none] (2) at (-2, 0) {...};
		\node [style=none] (3) at (2.25, 0) {...};
		\node [style=wirev] (4) at (0.75, 0) {};
		\node [style=none] (5) at (1.75, 0) {};
	\end{pgfonlayer}
	\begin{pgfonlayer}{edgelayer}
		\draw [style=diredge] (0.center) to (1);
		\draw [style=diredge] (1) to (4);
		\draw [style=diredge] (4) to (5.center);
	\end{pgfonlayer}
\end{tikzpicture} \ \ \mapsto \ \ \begin{tikzpicture}
	\begin{pgfonlayer}{nodelayer}
		\node [style=wirev] (0) at (0, 0) {};
		\node [style=none] (1) at (1.5, 0) {...};
		\node [style=none] (2) at (1, 0) {};
		\node [style=none] (3) at (-1, 0) {};
		\node [style=none] (4) at (-1.5, 0) {...};
	\end{pgfonlayer}
	\begin{pgfonlayer}{edgelayer}
		\draw [style=diredge] (3.center) to (0);
		\draw [style=diredge] (0) to (2.center);
	\end{pgfonlayer}
\end{tikzpicture}
      \quad\quad\quad\quad
    \begin{tikzpicture}
	\begin{pgfonlayer}{nodelayer}
		\node [style=wirev] (0) at (0, 0) {};
		\node [style=none] (1) at (1.5, 0) {...};
		\node [style=none] (2) at (1, 0) {};
		\node [style=none] (3) at (-1, 0) {};
		\node [style=none] (4) at (-1.5, 0) {...};
	\end{pgfonlayer}
	\begin{pgfonlayer}{edgelayer}
		\draw [style=diredge] (3.center) to (0);
		\draw [style=diredge] (0) to (2.center);
	\end{pgfonlayer}
\end{tikzpicture} \ \ \mapsto \ \ \begin{tikzpicture}
	\begin{pgfonlayer}{nodelayer}
		\node [style=none] (0) at (-1.5, 0) {};
		\node [style=wirev] (1) at (-0.5, 0) {};
		\node [style=none] (2) at (-2, 0) {...};
		\node [style=none] (3) at (2.25, 0) {...};
		\node [style=wirev] (4) at (0.75, 0) {};
		\node [style=none] (5) at (1.75, 0) {};
	\end{pgfonlayer}
	\begin{pgfonlayer}{edgelayer}
		\draw [style=diredge] (0.center) to (1);
		\draw [style=diredge] (1) to (4);
		\draw [style=diredge] (4) to (5.center);
	\end{pgfonlayer}
\end{tikzpicture} \]
\end{definition}
Two string graphs $G \sim G'$ are "semantically" the same and only differ by
the length of some of its wires. Note, that for any string graph $G$, its
wire-homeomorphism class has a unique minimal representative, which can be
obtained from $G$ by just contracting wires as much as possible, so
wire-homeomorphism is decidable.

In order to rewrite a string graph using a string graph rewrite rule, one first finds a matching of the LHS. Note we say an edge is incident to a subgraph $K \subseteq G$ if it connects a vertex in $K$ to a vertex in $G \backslash K$.
\begin{definition} \rm
  Let $L$ be a string graph with boundary $B \subseteq L$. Then a \textit{matching} of $L$ onto a string graph $H$ is an injective string graph homomorphism $m : L \to \widetilde H$ where $H \sim \widetilde H$ and the only edges incident to the image $m(L) \subseteq \widetilde H$ are also incident to $m(B)$.
\end{definition}
In other words, a matching satisfies the `no dangling wires' condition with respect to the boundary of $L$. Note that matching is done modulo wire-homeomorphism. This allows wires in the target graph to grow if necessary to injectively embed the pattern. In the example below, the embedding on the left fails, but the embedding into a wire-homeomorphic graph (right) succeeds:
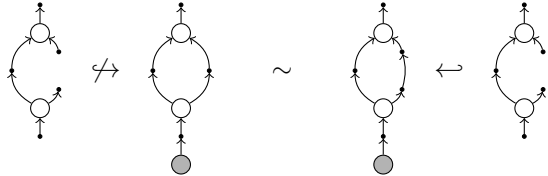
\[
\begin{tikzpicture}
	\begin{pgfonlayer}{nodelayer}
		\node [style=white dot] (0) at (-1.75, 1) {};
		\node [style=white dot] (1) at (-1.75, -1) {};
		\node [style=wirev] (2) at (-2.5, 0) {};
		\node [style=wirev] (3) at (-1.25, -0.5) {};
		\node [style=wirev] (4) at (-1.25, 0.5) {};
		\node [style=wirev] (5) at (-1.75, 1.75) {};
		\node [style=wirev] (6) at (-1.75, -1.75) {};
		\node [style=none] (7) at (0, 0) {$\not\hookrightarrow$};
		\node [style=wirev] (8) at (2, 1.75) {};
		\node [style=white dot] (9) at (2, 1) {};
		\node [style=wirev] (10) at (2, -1.75) {};
		\node [style=wirev] (11) at (1.25, 0) {};
		\node [style=white dot] (12) at (2, -1) {};
		\node [style=wirev] (13) at (2.75, 0) {};
		\node [style=gray dot] (14) at (2, -2.5) {};
	\end{pgfonlayer}
	\begin{pgfonlayer}{edgelayer}
		\draw [style=diredge] (6) to (1);
		\draw [style=diredge, in=-90, out=150, looseness=1.00] (1) to (2);
		\draw [style=diredge, in=-150, out=90, looseness=1.00] (2) to (0);
		\draw [style=diredge] (0) to (5);
		\draw [style=diredge, in=-105, out=30, looseness=1.00] (1) to (3);
		\draw [style=diredge, in=-30, out=105, looseness=1.00] (4) to (0);
		\draw [style=diredge] (10) to (12);
		\draw [style=diredge, in=-90, out=150, looseness=1.00] (12) to (11);
		\draw [style=diredge, in=-150, out=90, looseness=1.00] (11) to (9);
		\draw [style=diredge] (9) to (8);
		\draw [style=diredge, in=-90, out=30, looseness=1.00] (12) to (13);
		\draw [style=diredge, in=-30, out=90, looseness=1.00] (13) to (9);
		\draw [style=diredge] (14) to (10);
	\end{pgfonlayer}
\end{tikzpicture}
\qquad\sim\qquad
\begin{tikzpicture}
	\begin{pgfonlayer}{nodelayer}
		\node [style=white dot] (0) at (2, 1) {};
		\node [style=white dot] (1) at (2, -1) {};
		\node [style=wirev] (2) at (1.25, 0) {};
		\node [style=wirev] (3) at (2.5, -0.5) {};
		\node [style=wirev] (4) at (2.5, 0.5) {};
		\node [style=wirev] (5) at (2, 1.75) {};
		\node [style=wirev] (6) at (2, -1.75) {};
		\node [style=none] (7) at (0, 0) {$\hookleftarrow$};
		\node [style=wirev] (8) at (-1.75, -1.75) {};
		\node [style=gray dot] (9) at (-1.75, -2.5) {};
		\node [style=wirev] (10) at (-1.75, 1.75) {};
		\node [style=white dot, draw] (11) at (-1.75, 1) {};
		\node [style=wirev] (12) at (-2.5, 0) {};
		\node [style=wirev] (13) at (-1.25, -0.5) {};
		\node [style=wirev] (14) at (-1.25, 0.5) {};
		\node [style=white dot, draw] (15) at (-1.75, -1) {};
	\end{pgfonlayer}
	\begin{pgfonlayer}{edgelayer}
		\draw [style=diredge] (6) to (1);
		\draw [style=diredge, in=-90, out=150, looseness=1.00] (1) to (2);
		\draw [style=diredge, in=-150, out=90, looseness=1.00] (2) to (0);
		\draw [style=diredge] (0) to (5);
		\draw [style=diredge, in=-105, out=30, looseness=1.00] (1) to (3);
		\draw [style=diredge, in=-30, out=105, looseness=1.00] (4) to (0);
		\draw [style=diredge] (9) to (8);
		\draw [style=diredge, draw, in=-90, out=150, looseness=1.00] (15) to (12);
		\draw [style=diredge, draw, in=-150, out=90, looseness=1.00] (12) to (11);
		\draw [style=diredge, draw] (11) to (10);
		\draw [style=diredge, draw, in=-105, out=30, looseness=1.00] (15) to (13);
		\draw [style=diredge, draw, in=-30, out=105, looseness=1.00] (14) to (11);
		\draw [style=diredge, draw] (8) to (15);
		\draw [style=diredge, in=-75, out=75, looseness=1.00] (13) to (14);
	\end{pgfonlayer}
\end{tikzpicture}
\]
As usual, string graph rewrite rules are encoded as spans $L \leftarrow B \rightarrow R$, where $B$ is the common boundary of $L$ and $R$.
\begin{definition} \rm
  A \textit{rewrite} of a string graph $G$ by a rule $L \leftarrow B \rightarrow
  R$ using a matching $m : L \to \widetilde G$ (for $\widetilde G \sim G$)
  consists of a pushout complement (1) followed by a pushout (2) in
  \textbf{SGraph}:
  \ctikzfig{dpo-squares}
\end{definition}
It was shown in~\cite{open_graphs1} that for any matching of the LHS of a string graph rewrite rule, the pushout complement (1) and the pushout (2) exist and are unique, so DPO rewriting for string graph rewrite rules is well-defined.
\presec
\section{Encoded string graphs and B-ESG grammars}\label{sec:besg}
\postsec
We will introduce a type of context-free graph grammar suitable for defining families of string graphs, which is essentially a restriction of the class of B-edNCE grammars. However, to squeeze out a bit more expressive power, rather than using such grammars to generate string graphs themselves, we use them to generate \textit{encoded string graphs}. An encoded string graph allows us to `fold' some collection of fixed subgraphs into single edges, which will allow us more flexibility when it comes to the types of languages we can produce.
\begin{definition}[Encoded string graph] \rm
  Let $\mathcal E = \{ \alpha, \beta, \ldots \}$ be a finite set of \textit{encoding symbols}. An \textit{encoded string graph} is a string graph where we additionally allow edges labelled by encoding symbols $\alpha \in \mathcal E$ to connect pairs of node-vertices.
\end{definition}
\begin{definition}[Decoding system] \rm
  A \textit{decoding system} $T$ is a set of DPO rewrite rules of the form:
  \begin{equation}\label{eq:simple-dpo-def}
    \begin{tikzpicture}
	\begin{pgfonlayer}{nodelayer}
		\node [style=nodev] (0) at (-6.25, 0) {};
		\node [style=gray dot] (1) at (-4, 0) {};
		\node [style=none] (2) at (-5.25, 0.5) {$\alpha$};
		\node [style=none] (3) at (-2, 0) {$\Longrightarrow$};
		\node [style=gray dot] (4) at (3.25, 0) {};
		\node [style=nodev] (5) at (-0.25, 0) {};
		\node [style=none] (6) at (0.75, 0) {};
		\node [style=none] (7) at (2.25, 0) {};
		\node [style=none] (8) at (1.5, 0) {...};
	\end{pgfonlayer}
	\begin{pgfonlayer}{edgelayer}
		\draw [style=directed] (0) to (1);
		\draw [style=directed] (5) to (6.center);
		\draw [style=directed] (7.center) to (4);
	\end{pgfonlayer}
\end{tikzpicture}\ ,
  \end{equation}
  one for every triple $(\alpha, N_1, N_2) \in \mathcal E \times \mathcal N \times \mathcal N$, where the LHS consists of a single edge labeled $\alpha$ connecting an $N_1$-labelled node-vertex to an $N_2$-labelled node-vertex, and the RHS is a connected string graph that contains no inputs, outputs, or encoding labels.
\end{definition}
Note, the invariant part of~\eqref{eq:simple-dpo-def} consists of the two shared node-vertices. Thus, by construction $T$ is confluent (since no two rules apply in the same location) and terminating (since no encoding labels occur in the RHS of a rule). Thus, \textit{decoding} an encoded string graph consists of normalising with respect to $T$.

Throughout the rest of the paper, we assume that all of our grammars use the
same vertex and edge label alphabets, $\Sigma$ and $\Gamma$ respectively, and
also the same initial non-terminal $S$. The alphabet for terminal vertex labels
is $\Delta := \mathcal N \cup \mathcal W$.  We do not allow any non-final edge
labels. We assume $\mathcal E \subset \Gamma$ and any edge with label in
$\mathcal E$ will be called an \textit{encoding edge}.
\begin{definition}[B-ESG grammar]\label{def:besg} \rm
  A \textit{B-ESG} grammar is a pair $B = (G, T)$, where $T$ is a
  decoding system and
  $G=(\Sigma,\Delta,\Gamma,\Gamma,P,S)$ is a B-edNCE grammar, such that for
  every production
  $X \to (D,C) \in P$, the following conditions are satisfied:
  \begin{enumerate}
    \item[N1:] An edge carries an encoding label if and only if it connects a pair of node-vertices.
    \item[N2:] Any connection instruction of the form
      $(N, \alpha/\beta, x, d)$ where $N$ is a node-vertex label and $x$ is a node-vertex, must have $\beta \in \mathcal E$.
    \item[W1:] Every wire-vertex in $D$ has in-degree at most one and out-degree
      at most one.
    \item[W2:] There are no connection instructions of the form $(\sigma, \alpha/\beta, x, d)$ where $\sigma$ is any vertex label and $x$ is a wire-vertex.
    \item[W3:] For $W$ a wire-vertex label and each $\gamma$ and $d$, there is at most one connection instruction of the form $(W, \gamma/\delta, x, d)$ and we must have $\delta \not \in \mathcal{E}.$
    \item[W4:] Let $y$ be a non-terminal vertex with label $Y$ in $D$.
      If $y$ is adjacent to a wire-vertex labelled $W$ via an edge with direction $d$ and label
      $\beta$ or there's a connection instruction of the form
      $(W,\alpha/\beta,y,d)$, then all productions $Y$ must contain a connection
      instruction of the form $(W, \beta/\gamma,z,d)$.
  \end{enumerate}
\end{definition}
The conditions N1 and N2 guarantee that node-vertices never become directly connected by an edge, unless that edge has an encoding label. W1-3 ensure that wires never `split', i.e. wire-vertices always have at most one input or output. The final condition, which won't be necessary until the next section, ensures that inputs stay inputs and outputs stay outputs in a sentential form throughout the course of the derivation.
%
%
\begin{definition}[B-ESG concrete derivation] \rm
  A \textit{concrete derivation} for a B-ESG grammar $B = (G,T)$ with
  $S$ the initial non-terminal for $G$,
  consists of a derivation $S \Longrightarrow_*^G H_1$ in $G$, where $H_1$ contains no
  non-terminals, followed by a decoding $H_1 \Longrightarrow_*^T H_2$.
  We will denote such a concrete derivation as $S \Longrightarrow_*^G H_1
  \Longrightarrow_*^T H_2$ or simply with $S \Longrightarrow_*^B H_2$
  if the graph $H_1$ is not relevant for the context.
\end{definition}
Note that in the above definition $S$ refers to both the initial non-terminal
label, but also to a graph with a single vertex with label $S$ which is the
starting graph for a derivation. As usual, the \textit{language} of a B-ESG
grammar $B$ is given by $L(B):=\{H\ |\ S \Longrightarrow_*^{B} H\}$.
\begin{theorem}\label{lem:besg_language} \rm
  Every graph in the language of a B-ESG grammar is a string graph.
\end{theorem}
\begin{proof}
  Let $B=(G,T)$ be a B-ESG grammar. Decoding an encoded string graph will always
  produce a string graph, so it suffices to show that any derivation from $G$
  produces an encoded string graph. We call a sentential form an ESG-form if it is an
  encoded string graph, which possibly has some additional non-terminals that
  are either connected to node-vertices or are connected to wire-vertices in such
  a way that all wire-vertices have at most 1 in-edge and 1 out-edge. We show that any
  derivation starting from an ESG-form is an encoded string graph.
  This can be done by induction on the length of derivations. If the derivation
  is length 0, the sentential form has no non-terminals, so it is an encoded
  string graph. Otherwise, consider a derivation of length $n$. After the
  first step, any newly-introduced wire-vertices will have in-degree and out-degree
  at most 1 by W1 and W2, whereas the degrees of any already existing wire-vertices will
  not increase by W3. N1 and N2 will ensure that any resulting
  node-vertices will only be connected by edges with encoding labels, so the
  result is an ESG-form. Thus we can apply the induction hythothesis.
  Noting that $S$ is, in particular, an ESG-form completes the proof. \qed
\end{proof}
\begin{lemma}\label{lem:bounded_wire} \rm
  For every B-ESG grammar $B=(G,T)$, there exists $n \in \mathbb{N}$, such that
  if $H \in L(B)$ then $H$ does not contain a wire with size
  bigger than $n$.
\end{lemma}
\begin{proof}
  Let $n$ be the length of the longest wire in the bodies of the productions in
  $G$ and $T$, and consider an arbitrary sentential form obtained from $S
  \Longrightarrow_*^G H'$. From condition \textit{W2}, we see that expanding any
  of the non-terminals cannot create a new edge between an already established
  wire-vertex and a newly created wire-vertex. Therefore, a concrete derivation
  in $G$ will produce an encoded string graph with maximum length of any wire
  $n$. Then, while doing the decoding, $T$ will only replace edges between
  node-vertices and therefore a wire longer than $n$ cannot be established.
  \qed
\end{proof}
Naturally, we want to be able to decide if a given string graph is in a B-ESG
grammar. Since there should be no distinction between wire-homeomorphic string
graphs, we state the membership problem as follows:
\begin{problem}[Membership]
Given a string graph $H$ and a B-ESG grammar $B$, does
there exist a string graph $\widetilde H \sim H$, such that $\widetilde H \in L(B)$?
In such a case, construct a derivation sequence $S \Longrightarrow_* \widetilde
H$.
\end{problem}
\begin{theorem}\label{lem:member} \rm
  The membership problem for B-ESG grammars is decidable.
\end{theorem}
\begin{proof}
  First, we show that exact membership (i.e. not up to wire-homeomorphism) is
  decidable. From Theorem~\ref{lem:besg_language}, we know that any concrete B-ESG
  derivation produces an encoded string graph which is then decoded to a string
  graph. Since the decoding sequence can only increase the size of a graph, we
  can limit the problem to considering all graphs of size smaller than $H$.
  However, there are finitely many graphs whose size is smaller than $H$. For
  each such graph $H'$, we can then decide if $H' \in L(G)$ (this is the
  membership problem for B-edNCE grammars). Finally, we check if $H'
  \Longrightarrow_*^T H$ which is also clearly decidable. If no such graph $H'$
  exists, then the answer is no and otherwise the answer is yes.

  We now generalise to the wire-homeomorphic case. There may be infinitely many
  string graphs $\widetilde H$ such that $\widetilde H \sim H$, but by using
  Lemma~\ref{lem:bounded_wire}, it suffices to consider only those $\widetilde
  H$ which do not have wires longer than some fixed $n \in \mathbb{N}$. There
  are finitely many of these, so we can check if at least one of them is in
  $L(B)$ as before. Finally, since derivation sequences are recursively
  enumerable, if $\widetilde H \in L(B)$, we can also construct a concrete
  derivation sequence $S \Longrightarrow_* \widetilde H$.
  \qed
\end{proof}
In section~\ref{sec:transforming}, we will show how to use B-ESG grammars for rewriting. To do
this, we must show that the grammar produces some graphs which can be matched onto a given string graph, and
ideally that the set of \textit{all} matchings is finite, so that we can enumerate all of the relevant equalities.
This is not true in general, but it is true whenever the B-ESG grammar satisfies
some simple conditions:
\begin{definition}[Match-exhaustive B-ESG grammar]\label{def:non-bare} \rm
  We say that a B-ESG grammar $B = (G,T)$ is \textit{match-exhaustive}, if (1)
  any production $X \to (D,C)$ which contains a bare wire in $D$ has a
  finite bound on the number of times it can be expanded in any sentential form
  (2) no production contains an isolated wire-vertex in its body (3) there are
  no empty productions and (4) there are no productions consisting of a single node which is non-terminal.
\end{definition}
Note that condition (1) can be easily decided by examining all productions which
could possibly lead back to themselves via non-terminals, and seeing if they
contain any bare wires. In \cite{graph_grammar_handbook} it was furthermore shown
that any grammar can be transformed into an equivalent grammar satisfying
conditions (3) and (4).
\begin{problem}[Match-enumeration]
  Given a string graph $H$ and a B-ESG grammar $B$, enumerate
  all of the B-ESG concrete derivations $S \Longrightarrow_*^{B} K$, such that
  there exists a matching $m : K \to \widetilde H$ for some $\widetilde H \sim H$.
\end{problem}
\begin{theorem}\label{thm:match} \rm
  The match-enumeration problem for a B-ESG grammar $B$ is decidable if
  $B$ is a match-exhaustive grammar.
\end{theorem}
\begin{proof}
  Let $W$ be the number of wires in a string graph $H$. Then, for any
  $\widetilde H \sim H$, we know that $\widetilde H$ must have the same number of wires
  and node-vertices as $H$. However, the number of wire-vertices in $\widetilde H$ may be
  arbitrarily large.

  Condition (1) of Definition~\ref{def:non-bare} implies that there exists $n
  \in \mathbb{N}$, such that, for any $K \in L(B)$, $K$ has at most $n$ bare
  wires. Any matching of string graphs will map at most two non-bare wires onto
  a single wire. So, if $K$ has a matching on $\widetilde H$, then it can have at most $2W$
  non-bare wires. Therefore, $K$ can have at most $2W+n$ wires. From
  Lemma~\ref{lem:bounded_wire} we know that the length of any such wire
  is bounded. Condition (2) of Definition~\ref{def:non-bare} implies that
  any wire-vertex in $K$ is part of some wire, thus there is a bound on the
  number of wire-vertices in $K$ and we already know the number of node-vertices
  is also bounded. Thus, there are finitely many $K \in L(B)$ which could possibly
  have a matching onto some $\widetilde H \sim H$, so we can enumerate them.

  Finally, conditions (3) and (4) from Definition~\ref{def:non-bare} imply
  that the sentential forms of $B$ can only increase in size and
  therefore for any $K$ satisfying the above conditions there are finitely
  many concrete derivations $S \Longrightarrow_*^{B} K$ which we can
  enumerate.
  \qed
\end{proof}
\presec
\section{B-ESG rewrite patterns}\label{sec:rewrite}
\postsec
In the previous section we introduced a new type of grammar which generates
string graphs and which has some important decidability properties. Thus,
we can use a single B-ESG grammar to represent a single family of string graphs.
However, we still have not described how to encode equalities between families
of string graphs. This is the primary contribution of this section. We
introduce the notions of \textit{B-ESG rewrite pattern} and \textit{B-ESG pattern
instantiation} which show how this can be achieved in a formal way and then we
give examples of important equalities which cannot be encoded using the !-graph
formalism, but are expressible using B-ESG rewrite patterns.
\begin{definition}[B-ESG rewrite pattern]\label{def:besg_rewrite} \rm
  A \textit{B-ESG rewrite pattern} is a pair of B-ESG grammars $B_1 = (G_1, T)$
  and $B_2 = (G_2,T)$, where
  $G_1=(\Sigma,\Delta,\Gamma,\Gamma,P_1,S)$ and
  $G_2=(\Sigma,\Delta,\Gamma,\Gamma,P_2,S)$, such that
  there is a bijective correspondence between the productions
  in $P_1$ and $P_2$ given by $X \to (D_1,C_1) \in P_1$ iff
  $X \to (D_2,C_2) \in P_2$ and the corresponding pairs of productions
  satisfy the following conditions:
  \begin{enumerate}
    \item[NT:] There is a bijection between the non-terminal nodes in $D_1$ and
      the non-terminal nodes in $D_2$ which also preserves their labels.
    \item[IO:] There is a bijection between the inputs (resp. outputs) in $D_1$ and the
      inputs (resp. outputs) in $D_2$.
  \end{enumerate}
\end{definition}
Condition \textit{NT} ensures that we can
perform identical derivation sequences on both grammars $G_1$ and $G_2$ in the
sense that we can apply the same order of productions to corresponding non-terminal
vertices. This should become more clear from definition \ref{def:inst}.
When working with B-ESG rewrite patterns, without loss of generality, we will assume
that the corresponding inputs/outputs/non-terminal nodes are identified between
the two grammars by sharing the same name, instead of using a bijective function
explicitly.
\begin{definition}[B-ESG pattern instantiation]\label{def:inst} \rm
  Given a B-ESG rewrite pattern $(B_1, B_2)$, a B-ESG pattern instantiation is
  given by a pair of concrete derivations:
  \[S \Longrightarrow_{v_1,p_1}^{B_1} H_1 \Longrightarrow_{v_2,p_2}^{B_1} H_2
  \Longrightarrow_{v3,p3}^{B_1} \cdots \Longrightarrow_{v_n,p_n}^{B_1} H_n
  \Longrightarrow_*^T F  \]
  and
  \[S \Longrightarrow_{v_1,p_1}^{B_2} H_1' \Longrightarrow_{v_2,p_2}^{B_2} H_2'
  \Longrightarrow_{v3,p3}^{B_2} \cdots \Longrightarrow_{v_n,p_n}^{B_2} H_n'
  \Longrightarrow_*^T F' \]
\end{definition}
In other words, we use an identical derivation sequence in the two B-edNCE grammars
to get two encoded string graphs, which are then uniquely decoded using the
productions of $T$. These ideas are similar to the \textit{pair grammars} approach
presented in \cite{pair_grammars}, but different in that we are using a more
general notion of grammar, our extension to the context-free grammars is more
limited and our focus is on string diagram reasoning rather than computer
program representation. These ideas have been further generalised in
\cite{triple_grammars}.
\begin{theorem}\label{lem:inst} \rm
  Every B-ESG pattern instantiation is a string graph rewrite rule.
\end{theorem}
\begin{proof}
  Consider a concrete derivation as in Definition~\ref{def:inst}.
  From Theorem~\ref{lem:besg_language}, we know that $F$ and $F'$ are string graphs.
  We have to show that they have the same set of inputs/outputs. Note, that the
  decoding process using the productions of $T$ cannot establish any new inputs or
  outputs and thus we can reduce the problem to showing that $H_n$ and
  $H_n'$ have the same sets of inputs/outputs.

  Setting $H_0:=S:=H_0'$, we can prove by induction that any pair of sentential
  forms $(H_i, H_i')$, has the same set of inputs. This is trivially true for
  $(H_0, H_0')$.  Assuming that $(H_k,H_k')$ have the same set of inputs,
  observe that $(H_{k+1}, H_{k+1}')$ is obtained by applying a corresponding
  pair of productions to $(H_k, H_k')$. These productions satisfy all of the
  listed conditions in Definition~\ref{def:besg} and
  Definition~\ref{def:besg_rewrite}. In particular, conditions \textit{W3} and
  \textit{W4} guarantee that the in-degree of all wire-vertices in both $H_k$
  and $H_k'$ won't be affected.  Condition \textit{W2} implies that the newly
  created wire-vertices in $H_{k+1}$ and $H_{k+1}'$ are not connected to any
  of the previously established vertices in $H_k$ and $H_k'$ respectively.
  Combining this with condition \textit{IO} ensures that any newly created
  wire-vertices are inputs in $H_{k+1}$ iff they are inputs in $H_{k+1}'$.
  Thus, $H_{k+1}$ and $H_{k+1}'$ have the same set of inputs. In particular,
  $H_n$ and $H_n'$ have the same inputs. By symmetry, $H_n$ and $H_n'$ have the same output vertices.
  \qed
\end{proof}
As mentioned in the introduction, we now gain previously non-existent expressive
power for families of string graph rewrite rules. For instance, we can write a
rule that merges a chain of white node-vertices, each with 1 input, into a
single white node-vertex with $n$ inputs as follows:
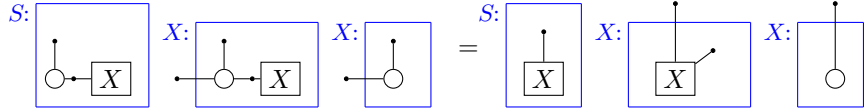
\ctikzfig{line-example-rule}
In the introduction, we also gave an example of a rule involving clique-like
graphs. In fact, a minor modification to~\eqref{eq:complete-rule-example}
yields a rule that is directly relevant for quantum computation. The following
rule, called \textit{local complementation}:
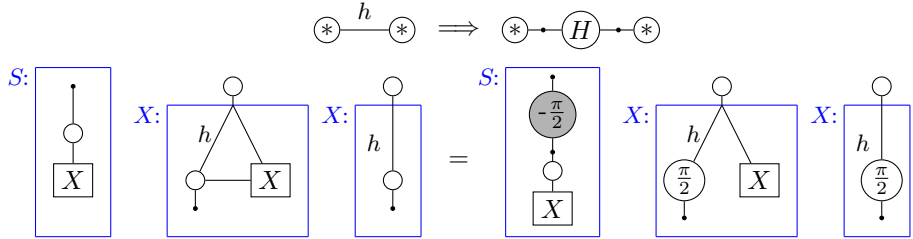
\ctikzfig{local-comp-B-ESG}
is crucial to establishing the completeness of a rewrite system called the
\textit{ZX-calculus}, which has been used extensively for
reasoning about quantum circuits and measurement-based quantum computation. This
rule uses a non-trivial decoding system, where the $h$-labelled edges are interpreted as wires containing a single node-vertex.
Here we are using `$*$' to mean `any node-vertex label'. The node-vertices
labelled $\pm \frac\pi2$ represent \textit{quantum phase gates}, whereas
the $H$ is a \textit{Hadamard gate}. The interested reader can find a
detailed description of the ZX-calculus in e.g.~\cite{bigmainzx} the
completeness theorem based on local complementation in~\cite{Backens:2012fk}.
\presec
\section{Transforming B-ESG grammars}\label{sec:transforming}
\postsec
In the previous section, we showed how a B-ESG rewrite pattern can encode
(infinitely) many equalities between string diagrams, where we represent the
equalities as string graph rewrite rules. Thus, we can use B-ESG rewrite
patterns to encode axioms and axiom schemas of a string diagram theory.
However, in order to perform equational reasoning between families of string
diagrams, we need to show how we can use such axioms in our formalism in order
to obtain new equalities. This is the primary contribution of this section. In
particular, we show how we can transform a B-ESG rewrite pattern into another
one using a string graph rewrite rule in an admissible way. Finally, we
generalise this result by showing how we can transform B-ESG rewrite patterns
using another B-ESG rewrite pattern in an admissible way.
\begin{definition}[Final subgraph] \rm
  Given a production $X \to (D,C)$ in a B-edNCE grammar, we say that
  a subgraph $S$ of $D$ is final if $S$ is a string graph which is not
  adjacent to any non-terminals in $D$ and no vertex of $S$ has associated
  connection instructions in $C$.
\end{definition}
\begin{definition}[B-ESG transformation step] \rm
  We say that a given grammar $B = (G,T)$ can be transformed
  into another grammar $B'=(G',T)$ using a given string graph rewrite rule $SR = L \leftarrow I
  \rightarrow R$, if (1) there exists a production $p = X \to (D,C)$ in $G$,
  such that $L$ can be matched into a final subgraph of $D$ (2) $G'$ has the
  same productions as $G$, except for $p$ which has been modified to
  $p' := (D',C)$ where $D'$ is the result of applying the rewriting rule $SR$ to
  the matching of (1).
\end{definition}
\begin{theorem}\label{lem:transform} \rm
  Given a B-ESG grammar $B$, the pair $(B,B')$ is a B-ESG rewrite pattern,
  where $B'$ is the result of a B-ESG transformation step applied to $B$.
\end{theorem}
\begin{proof}
  Let $B'=(G',T)$. Assume that the
  modified production is $p = X \to (D,C)$ of $B=(G,T)$ and the corresponding
  production in $G'$ is $p' = X \to (D',C)$. Rewriting $D \leadsto D'$ only
  affects a final subgraph $S$ of $D$ and thus it cannot establish any edges between
  non-terminals or modify the connection instructions. Therefore, $G'$ is a
  B-edNCE grammar and conditions \textit{N2,W2,W3,W4} and \textit{NT} are satisfied.
  Conditions \textit{W1} and \textit{N1} are also
  preserved, because we are replacing a string graph $S$ with another string graph
  and the edges between vertices in $D\backslash S$ and $S$ are preserved. Thus,
  $B'$ is a B-ESG grammar. Finally, condition \textit{IO} is satisfied, because
  string graph rewriting does not create or remove inputs/outputs.
  \qed
\end{proof}
\begin{theorem}\label{thm:admisible} \rm
  Let $B$ be a B-ESG grammar, let $SR$ be a string graph rewrite rule
  and let $(B,B')$ be the B-ESG rewrite pattern
  induced by $B$ and $SR$ as per Theorem~\ref{lem:transform}. Then, $(B,B')$ is admisible in the sense that any pattern
  instantiation $S \Longrightarrow_*^B F$ together with
  $S \Longrightarrow_*^{B'} F'$ are such that the string graph $F'$ can
  be obtained from $F$ via repeated applications of $SR$.
\end{theorem}
\begin{proof}
  Let $p$ be the production of $B$ which is modified by $SR$ and let $p'$ be
  the result of the modification. Because the modification is done on a
  final subgraph $S$ of $p$, as $p$ is expanded, a copy of $S$ is created
  which is not adjacent to any previously established vertices. Moreover,
  the copy $S$ is not adjacent to any non-terminals and thus its
  neighbourhood won't change in the following sentential forms. Therefore,
  to obtain $F'$ from $F$, we have to apply the rule $SR$ exactly $n$ times,
  where $n$ is the number of times the production $p$ appears in the
  pattern instantiation.
  \qed
\end{proof}
\begin{corollary} \rm
  If $B$ is a B-ESG grammar and $(B_1, B_2)$ is a B-ESG
  rewrite pattern with $B_1$ match-exhaustive, then we can enumerate the pattern instantiations of $(B_1, B_2)$
  which induce an admisible  rewrite pattern $(B, B')$.
\end{corollary}
\begin{proof}
  From Theorem~\ref{lem:inst}, any pattern instantiation of $(B_1, B_2)$ is a
  string graph rewrite rule. If such a string graph rewrite rule matches a
  final subgraph in some production of $B$, then it induces an admisible rewrite
  pattern $(B,B')$ as shown in Theorem~\ref{thm:admisible}. Finally, by
  Theorem~\ref{thm:match}, we can enumerate all of those instantiations and thus
  all such $B'$.
  \qed
\end{proof}
\presec
\section{Conclusion and future work}
\postsec
In this paper, we have defined a family of grammars that is suitable for
describing languages of string graphs and string graph rewrite rules. We have
also showed that these languages admit nice decidability properties. There are two
natural directions in which to extend this work. Firstly, the conditions of a
B-ESG grammar are sufficient, but not necessary to obtain a language consisting
of only string graphs. One might ask if there exist other natural and easily-decidable
conditions yielding such languages, and how those conditions relate to the B-ESG
conditions.

The second, and more compelling direction for future work is the
development of tools for reasoning with B-ESG rewrite patterns and, most
importantly, deriving new patterns. We took the first step in this direction in
Section~\ref{sec:transforming}, where we described a fairly limited technique
whereby grammars can be transformed by using string graph rewrite rules to
rewrite the final parts of some productions. However, in the case string graph
patterns based on !-boxes, which we briefly discussed in the introduction, we
additionally have the ability to do geniune pattern-on-pattern rewriting, and
derive new !-box rules using !-box induction. We expect both of these techniques
to extend to the context-free case. Thus, we would ultimately like a much richer
notion of rewriting, whereby B-ESG patterns can be used to rewrite non-final
parts of grammars, where e.g.~non-terminals are allowed to match on non-terminals,
subject to suitable consistency conditions. Similarly, we intend to
extend !-box induction to a more general structural induction principle that can
be used to derive B-ESG rewrite patterns from basic string graph rules. This will
provide a method for deriving infinite families of rules which is both extremely
powerful and capable of producing machine-checkable diagrammatic proofs.

\noindent \textbf{Acknowledgements.} We would like to thank the anonyomous
reviewers for their feedback. We also gratefully acknowledge
financial support from EPSRC, the Scatcherd European Scholarship, and the John
Templeton Foundation.
\bibliographystyle{plain}
\bibliography{vladimir_refs}
\postsec
\def\baselinestretch{1.0}

\end{document}